\documentclass[11pt]{elsarticle}

\usepackage{color}

\usepackage{amsmath,amsthm}
\usepackage{lineno,hyperref}
\modulolinenumbers[5]

\journal{Operations Research Letters}
\usepackage{times}

\newtheorem{lemma}{Lemma}[section]
\newtheorem{theorem}[lemma]{Theorem}

\newtheorem{corollary}[lemma]{Corollary}

\newcommand{\eps}{\epsilon}

\definecolor{forest}{rgb}{0, .6, 0}
\definecolor{darkmagenta}{rgb}{.6, 0, .6}

\begin{document}
\begin{frontmatter}
\title{Symmetric Linear Programming Formulations for Minimum Cut with  Applications to TSP}




\author[bobaddress]{Robert D. Carr\footnote{Work done while at  Sandia National Laboratories. This material is based upon research supported in part by the U. S. Office of Naval Research under award number N00014-18-1-2099.}}
\author[jenaddress]{Jennifer Iglesias\footnote{Work done while at Carnegie Mellon University. This material is based upon research supported in part by the National Science Foundation Graduate Research Fellowship Program under Grant No. 2013170941.}}
 \author[giuseppeddress]{Giuseppe Lancia} 
 \author[benaddress]{Benjamin Moseley\footnote{Work done while at Washington University in St. Louis. This matieral is based upon research supported in part by a Google research Award and NSF Grants CCF-1824303, CCF-1830711 and CCF-1733873}}

\address[bobaddress]{University of New Mexico}
\address[jenaddress]{Waymo}
\address[giuseppeddress]{D.I.M.I., University of Udine, Italy}
\address[benaddress]{Carnegie Mellon University}





\begin{abstract}
We introduce multiple symmetric LP relaxations for minimum cut problems. The relaxations give optimal and approximate solutions when the input is a Hamiltonian cycle.  We show that this leads to one of two interesting results.  In one case, these LPs always give optimal and near optimal solutions, and then they would be the smallest known symmetric LPs for the problems considered.  Otherwise,  these LP formulations give strictly better LP relaxations for the traveling salesperson problem than the subtour relaxation. We have the smallest known LP formulation that is a $\frac{9}{8}$-approximation or better for min-cut. In addition, the LP relaxation of min-cut investigated in this paper has interesting constraints; the LP contains only a single typical min-cut constraint and all other constraints are typically only used for max-cut relaxations. 
\end{abstract}

\begin{keyword}
Minimum cut \sep Linear Program \sep Approximation \sep Traveling Salesperson Problem


\end{keyword}

\end{frontmatter}

\section{Introduction}

The minimum cut problem is one of the most fundamental problems in computer science and has numerous applications in other fields.  Consequently, there has been a vast amount of research on the topic (see \cite{ChekuriGKLS97} for an overview).   In this problem you are given an undirected graph $G$ on $n$ nodes and an undirected non-negative cost function $c$ on \emph{every} pair of nodes denoting the capacity of the edges in the graph.  In the global cut problem the goal is to find a non-empty set $S \subset V = \{1, 2, \dots, n\}$ where $S \neq V$ such that the sum of the edge weights crossing the cut $(S, V\setminus S)$ is minimized.  In the $s,t$ cut problem it must be the case that $s \in S$ and $t \in V \setminus S$ where $s,t \in V$ are given.  In the fixed sized partition minimum cut problem, the set $S$ must have size $\alpha$ where $\alpha$ is an input parameter.  The $s,t$ cut and global minimum cut problems are polynomially time solvable, but the fixed size partition minimum cut problem is known to be NP-Hard.

It is well known that the minimum cut problem yields polynomial time algorithms and a variety of efficient algorithms are known \cite{ChekuriGKLS97,Karger00} as well as efficient parallel algorithms \cite{KargerS96,LattanziMSV11}.  One widely used technique is mathematical programming.  A few integer linear programs are known for the minimum cut problems \cite{ChekuriGKLS97, carr2009compacting} and for several of them the LP relaxation  gives exact solutions to the minimum cut problem (global or $s,t$).  Like most problems, finding a linear programming relaxation for a given problem that is different than the obvious relaxation is generally a challenging task.   The smallest known linear programming relaxation of the global mincut problem,   called the $w$-LP, was given in \cite{carr2009compacting}.  This formulation has $O(n^2)$ variables and $O(n^3)$ constraints.  From a theoretical viewpoint, it is interesting to determine the bounds of the smallest linear program for a given problem, since small linear programs can lead to more efficient algorithms in practice.  For the minimum cut problem, finding a small relaxation not only can improve the performance for finding the minimum cut of a graph, but 
can be also instrumental to the solution  of other fundamental problems. 

It was shown by  in \cite{Yannakakis91} that no symmetric LP can have polynomial size for TSP or perfect matching.  Roughly speaking, an LP is said to be symmetric if the nodes in a graph can be permuted without changing the feasible region.  Further, it is known that certain classes of LP's can be lifted to give symmetry \cite{Yannakakis91}. This was originally thought to be a way to show that not every problem in P has a polynomial size LP. 
In particular, Yannakakis proved also that there are no compact LP formulations to perfect matching that
have symmetricity, and he argued that the lack of existence of a compact LP formulation to perfect matching that has symmetricity is strong evidence there is no compact LP formulation to
perfect matching (even asymmetric).  Indeed, the underlying perfect matching problem has
symmetricity.   Thus, if there were such an asymmetric LP formulation then
intuitively one could likely lift it to a slightly larger LP formulation that does have symmetricity.  But this
contradicts Yannakakis' result.
It turns out that Yannakakis was correct, and
recently, \cite{Rothvoss17} showed that the matching polytope does not have a polynomial size linear program.

\bigskip
\noindent \textbf{Motivation:}  There is a considerable amount of work on linear programming formulations.  For problems which are in  {\bf P}, typically non-exact linear programming formulations are only of interest if they are very small.  Small non-exact linear programs can be interesting because, in many cases, they can be used to give fast approximate guarantees on the optimal solution's value.  

In the Traveling Salesperson Problem (TSP), a salesperson wants to visit a set of cities and return home~\cite{junger1995traveling, goemans1995worst}.  There is a cost $c_{ij}$ of traveling from city $i$ to city $j$, which is the same in either direction for the Symmetric TSP.  The objective is to visit each city exactly once, minimizing total travel costs. The subtour elimination linear program is a natural and easily solvable linear programming relaxation of TSP:
\begin{displaymath}
\begin{array}{llll}
\mbox{minimize}&c\cdot x \\
\mbox{subject to} 
&\sum_{i} x_{ij} = 2 & \forall j \in V \\
&\sum_{ij \in \delta(S)} x_{ij} \geq 2 & \forall \emptyset \neq S\subset V  \\
&x_{ij} \geq 0 & \forall ij \in E(V).
\end{array}
\end{displaymath}
There is no natural compact LP relaxation for TSP which is stronger than the subtour elmination linear program relaxation.  The TSP uses the minimum cut as a subroutine \cite{carr2009compacting}. 

Despite most work focusing on finding small non-exact linear programming formulations, there is another property that non-exact linear programming formulations can have for problems in {\bf P} which is of interest to study.   Consider the minimum cut problem.  Say that one finds a linear programming formulation which returns a non-exact solution, yet one can prove that the LP's solution is \emph{exact} when the cost function (input graph) is a Hamiltonian cycle.   Although the LP is non-exact for general graphs, we actually can use the LP to improve the subtour relaxation for the TSP problem.  Indeed, we show in this paper,  if the LP relaxation were {\em not} an exact formulation of the minimum cut problem (but exact on Hamiltonian cycles), it can be used to strengthen the subtour LP relaxation for the TSP problem. 

We call this feature of the LP \emph{interesting}.  Interesting LPs are not limited to the relationship between minimum cut and TSP.  For example, suppose we had a compact relaxation of the minimum directed cut problem, but not necessarily an exact formulation. Further assume that this relaxation gives the exact answer of 1 when the cost function is an arborescence. Then, either this is an exact formulation of  directed cut or we can strengthen the LP relaxation for the directed arborescence problem. 

This leads to a somewhat unintuitive situation. One would want to find an LP formulation that is exact for computing a minimum cut when the graph is a Hamiltonian cycle, but simultaneously as loose as possible for computing the minimum cut for any other graph.  In this case, one can strengthen the subtour LP relaxation for the TSP problem by the largest amount.  Indeed, any time the minimum cut LP gives a non-exact solution this is a certificate that the objective function is outside the TSP polytope, even if this solution were in the subtour TSP relaxation.   We can then cut off this point in a new TSP relaxation, thus making it better than the subtour relaxation.

\bigskip
\noindent \textbf{Results.}  In this work, we consider linear programming relaxations for the fixed size partition and global minimum cut problems.  We begin by giving a symmetric relaxation assuming one knows $\alpha := |S|$, the size of one side of the cut.  We will call this the $\alpha$-LP.  Note that one can assume $\alpha \leq n/2$ since one side of the cut must have at most $n/2$ vertices. Our main results are the following:

\begin{itemize}
\item The $\alpha$-LP gives the optimal solution for the fixed size partition minimum cut problem when the input graph is an unweighted Hamiltonian cycle for any $\alpha \leq n/2$ (has the Hamiltonian cycle property).
\item The $\beta$-LP is a relaxation of the $\alpha$-LP which also has the Hamiltonian cycle property.
\item The $\gamma$-LP is a further relaxation of the $\alpha$-LP which is within $\frac{8}{9}$ of the $\alpha$-LP, but creates a smaller set of LPs which need to be solved to approximate the minimum cut.
\end{itemize}

Now, we do not know for general graphs if the $\alpha$-LP returns at least that of the minimum global cut.  However, if this is the case, then by extending the $\alpha$-LP to include all possible values of $\alpha$ this will be the smallest known LP for the global minimum cut problem that is symmetric.  The extension of the $\alpha$-LP to include all values of $\alpha$ has  $O(n^3)$ variables and $O(mn^2)$ constraints where $m$ is the size of the support of the cost function.   This compares to the previous best known formulation with $O(mn^2)$ variables and $O(mn^2)$ constraints which can be obtained by extending the standard formulation of $s,t$ minimum cut. Beyond being interesting from a theoretical viewpoint, symmetric LPs are also typically easier to analyze than other LPs.  

Alternatively, say that the $\alpha$-LP admits feasible solutions of value smaller than the minimum global cut.  In this case, we can use the $\alpha$-LP to obtain new valid inequalities for the TSP.  Here, by leveraging the fact that the $\alpha$-LP is optimal for Hamiltonian cycle graphs, we can show that this LP can be exploited to strengthen the standard subtour relaxation for TSP.  We show this by using compact optimization methods. Thus, we have a favorable situation, i.e.,  this new formulation either gives one of two interesting conclusions.

We then relax the $\alpha$-LP further and create a new model we call the $\beta$-LP. While the $\alpha$-LP assumes $\alpha$ is an integer, in the $\beta$-LP we relax this assumption by allowing $\alpha$ to be the convex combination of two consecutive integers. We show that the $\beta$-LP has the same Hamiltonian cycle property as the $\alpha$-LP even though the $\beta$-LP is a relaxation. We then show that in some special cases that the $\alpha$-LP and $\beta$-LP are lower bounded by the minimum cut. 

Lastly, we extend the $\beta$-LP relaxation and give another relaxation which  is $1/(1+\eps)$-approximation for the minimum cut problem on Hamiltonian cycle graphs and a $\frac{8}{9}$-approximation for the minimum cut problem when $\epsilon = 1$. This new relaxation has $O(n^2\log_{1+\eps} n)$  variables and   $O(mn\log_{1+\eps} n)$  constraints for any $\eps >0$.  We obtain this by extending the $\beta$-LP to express $\alpha$ as the convex combination of two integers, $\gamma$ and $\lceil (1+\eps)\gamma \rceil$. The $\beta$-LP has size  $O(n^2)$ variables by $O(mn)$ constraints.  By allowing variable sizes in the partitions where $\alpha$ is in the range $\gamma \leq \alpha \leq 2\gamma$, we can geometrically choose a logarithmic number of values for $\gamma$. For these values, we bound the objective of the $\gamma$-LP  by $16/9$ on Hamiltonian cycle graphs. Although  this LP is not exact for the minimum cut problem on Hamiltonian cycle graphs, one can still use this relaxation to potentially improve the solution for the TSP relaxation, while yielding a significantly smaller LP.

\section{Preliminaries}

We are going to construct symmetric LPs that solve or approximate fixed sized partition and  
(global) minimum cut.    
A cut in an undirected graph $G=(V,E(V))$ partitions the vertex set $V$ of 
$n$ vertices into a set of 
top nodes and a set of bottom nodes, with the edges in the cut going between 
these two sets.  Given a vector $c_{ij}\geq 0$ of capacities 
and edge variables $x_{ij}$ indicating whether edge $\{i,j\}$ is in a 
cut for $\{i,j\}\in E(V)$ we want to minimize $c\cdot x$.  Minimum $s$-$t$ cut 
where $s$ is specified to be on the bottom and $t$ is specified to be on the 
top is the most usual type of minimum cut problem (without such a specified 
pair of nodes it is the minimum global cut problem). 

For the definition of a symmetric linear program, we use the definition given in \cite{Yannakakis91}.  Let $\pi$ denote a permutation of the nodes.  In the complete graph, a permutation of the nodes defines a mapping of the costs of the edges of the graph such that each cost $c_{ij}$ is mapped into  $c_{\pi(i)\pi(j)}$.  Let $P(x,y)$ denote a polytope over the variables $x_{ij}$ on the edges of a graph and other variables $y$.   A polytope $P$ of a linear program is said to be \emph{symmetric} if for all permutations $\pi$ of the nodes the new variables $y$ can be extended so $P$ remains invariant.  A linear program can be seen to yield a symmetric polytope if renaming the variables for the nodes yields the same linear program.  Note that, however, a linear program need not necessarily have this property to be symmetric.

Minimum $s$-$t$ cut has 
an easy (naturally integer) LP formulation, shown below.  Let node variables $h_i$ 
indicate whether $i$ is on top ($h_i=1$)  or on bottom ($h_i=0$) in a cut.   
Furthermore, let $x_{ij}$ be edge variables indicating which edges are in the cut.
\begin{displaymath}
\begin{array}{llll}
\mbox{minimize}&c\cdot x \\
\mbox{subject to} 
&x_{i,j} \geq h_j - h_i& \forall ij\in E(V) \\
&x_{i,j} \geq h_i - h_j& \forall ij\in E(V) \\
&0=h_s \leq h_i \leq h_t = 1& \forall i \in V.
\end{array}
\end{displaymath}

One can also formulate the NP-hard maximum cut problem.  A naive IP 
formulation and LP relaxation is 
\begin{displaymath}
\begin{array}{llll}
\mbox{maximize}&c\cdot x \\
\mbox{subject to} 
&x_{i,j} \leq h_j + h_i& \forall ij\in E(V) \\
&x_{i,j} \leq 2 - h_j - h_i& \forall ij\in E(V) \\
&0\leq h_i \leq 1, \, h \mbox{ integer}& \forall i \in V.
\end{array}
\end{displaymath}
However, setting $h=1/2$ shows that this is a weak LP relaxation.  

Suppose one constrained minimum cut so that it has a fixed sized partition with exactly 
$\alpha \leq n/2$ nodes on top.  Note that one can always assume $\alpha \leq n/2$ since one side of the cut always at most $n/2$ nodes.  Call this the minimum $\alpha$-cut problem, 
which is also an NP-hard problem \cite{GareyJ1979}.  We note that when $\alpha = n/2$ this is equivalent to the minimum edge bisection problem.  We will now construct an LP relaxation 
for this problem.  Let $\delta(i)$ be the edges adjacent to $i$ in $E$ and $x(\delta(i)) = \sum_{ij \in \delta(i)} x_{i,j}$.  For $i \in V$ consider $x(\delta(i))$. Suppose $h_i=0$.  
Then there are $\alpha$ nodes opposite $i$, so $x(\delta(i))=\alpha$.  Now, 
suppose $h_i=1$.  Then there are $n-\alpha$ nodes opposite $i$, so 
$x(\delta(i))=n-\alpha$.  Hence, it is always the case that 
\begin{displaymath}
x(\delta(i)) = (n-2\alpha)h_i + \alpha.
\end{displaymath}
Also, for distinct $i,j,k\in V$ we have the (metric) triangle inequality 
property that $x_{i,j} \leq x_{i,k} + x_{k,j}$.  

Here is the first LP relaxation: 
\begin{align}
\mbox{minimize\ } & c\cdot x^\alpha  \label{first-lp} \\
\mbox{subject to}  
 \sum_{i = 1}^n h^\alpha_i &=  \alpha \nonumber \\
x^\alpha(\delta(i)) &=    (n - 2\alpha)h^\alpha_i + \alpha&\forall i\in V \nonumber \\
x^\alpha_{i,j} &\leq    h^\alpha_i + h^\alpha_j &\forall \{i,j\} \in E \nonumber \\ 
x^\alpha_{i,j} &\leq   x^\alpha_{k,i} + x^\alpha_{k,j} &\mbox{ distinct }i,j,k\in V \nonumber\\
  x^\alpha,h^\alpha &\geq  0 \nonumber . 
\end{align}
The above relaxation is stronger than we need for the results in this paper, we will refer to it as the old $\alpha$-LP. It can be relaxed to the following, which we call the (new) $\alpha$-LP:
\begin{align}
\mbox{minimize\ } & c\cdot x^\alpha  \label{symlp}\\
\mbox{subject to}  
 \sum_{i = 1}^n h^\alpha_i &\leq  \alpha \label{constalp:one} \\
 \sum_{ij \in E} x^\alpha_{i,j} & \geq   \alpha(n-\alpha) & \label{constalp:two} \\
x^\alpha_{i,j} &\leq    h^\alpha_i + h^\alpha_j &\forall \{i,j\} \in E \label{constalp:three}\\ 
x^\alpha_{i,j} &\leq   x^\alpha_{k,i} + x^\alpha_{k,j} &\mbox{ distinct }i,j,k\in V : c_{j,k}>0 \label{constalp:four}\\
x^\alpha,h^\alpha &\geq  0. \nonumber
\end{align}
We only need the triangle inequalities $x_{i,j} \leq x_{i,k} + x_{k,j}$ when $c_{j,k} > 0$ in our proof of Theorem~\ref{thm:alphaLPHamCycle}. The remaining triangle inequalities are unused. We obtain constraint~\eqref{constalp:two} by summing the constraints giving $x(\delta(i))$ in the previous LP. We emphasize that this $\alpha$-LP has size $O(mn)\times O(n^2)$ as we only need a subset of the triangle inqualities.

Denote the $\alpha$-LP relaxation's feasible region by 
\begin{displaymath}
A^{\alpha}z^{\alpha} \geq b^{\alpha},z^{\alpha} \geq 0,
\end{displaymath}
where $z^{\alpha} = [x^\alpha,h^\alpha]$. The $\alpha$-LP constraint~\eqref{constalp:two} is the only constraint which typically appears in min-cut relaxations while constraints~\eqref{constalp:one},~\eqref{constalp:three}, and~\eqref{constalp:four} generally appear in max-cut linear programming relaxations. 

In the next section we will bound the objective value of this LP.  Before we do that, we introduce the $w$-LP of \cite{carr2009compacting} which the LP can be compared to.  In this linear program $w_k = 1$ when $k$ is the last node (numerically) on top in the cut (where, w.l.o.g., we can assume that node $n$ is always on bottom).

\begin{equation}\label{soda}
\begin{array}{lrlll}
\mbox{minimize\ } c\cdot x \\
\mbox{subject to} \; 
&x_{i,j} + 2w_k &\leq& x_{k,i} + x_{k,j}&\forall k<i<j \\
& x_{ki} &\geq& w_k &\forall k<i \\
& \sum_{k=1}^{n-1}w_k &=& 1,\\
&0\leq  x,w &\leq& 1. 
\end{array}
\end{equation} 

While the $w$-LP is asymmetric, the $\alpha$-LP is symmetric.

\section{Hamiltonian Cycle Property}
Denote the convex hull of the set of the incidence vectors of all the cuts of
a complete graph $G=(V,E)$ by $Q$.  An LP whose feasible region $P\supset Q$
that minimizes or maximizes an objective function given by edge coefficients
$c$ is a relaxation of the cut problem.  If the LP is meant to bound the
{\it minimum} cut of $G$ it is said to be a relaxation of the {\it minimum} cut
problem.  We introduce the {\it dominant} polyhedron of a polytope.  The
polyhedron $dom(Q)$ is a set of points $y$ such that there exists an $x\in Q$
such that $y\geq x$.  An LP relaxation of
the minimum cut problem is an LP mincut formulation if its feasible region
$P \subset dom(Q)$.

We now introduce a surprising principle for creating LP relaxations of the TSP.
Consider {\it any} linear program
\begin{equation} \label{anylp}
  \begin{array}{lrll}
    \mbox{minimize }&c\cdot x \\
    \mbox{subject to} \\
    &A\cdot x &\geq& b \\
    &x &\geq& 0
  \end{array}
\end{equation}
on the edge variables $x$ of a complete graph (and possibly other variables).
If for any Hamilton cycle, when its incidence vector is plugged into $c$ the
minimum of this LP is a constant $\eta > 0$, then this LP is said to satisfy
the {\it Hamilton cycle property}.  If this LP is symmetric on the $x$ edge 
variables and the nodes, its minimum $\eta$ will automatically be a constant
for all Hamilton
cycles.  If this constant is strictly positive, the LP will then satisfy the
Hamilton cycle property.  An LP that satisfies the Hamilton cycle
property gives rise to an LP relaxation of the TSP through duality methods.
Let $\hat{c}$ be the costs for the TSP.  Let $c$ be the variables of the TSP
relaxation, which are so named to relate to the cost function of the LP
satisfying the Hamilton
cycle property.  The TSP relaxation is given by
\begin{equation} \label{tsprelax}
  \begin{array}{lrlll}
    \mbox{minimize }&\hat{c}\cdot c \\
    \mbox{subject to } \\
    &y\cdot A &\leq& c \\
    &y\cdot b &\geq& \eta \\
    &c(\delta(i)) &=& 2 &\forall i\in V \\
    &y,c &\geq& 0
  \end{array}
\end{equation}

With respect to the subtour elimination relaxation of the LP, the idea of LP~\eqref{tsprelax} is to keep the degree constraints, but replace
the subtour elimination constraints with the dual of LP~\eqref{anylp}.
Indeed, if the objective function $c^*$ plugged into LP~\eqref{anylp} satisfying the
Hamilton cycle property has a minimizer $x^*$ with objective value strictly
less than $\eta$, this LP solution is a certificate that $c^*$ is not in the
TSP polytope, 
and we say $c^*$ is {\em certified by LP~\eqref{anylp} not to be in the
TSP polytope}.  Moreover, $(y,c^*)$ is not feasible 
for any $y \geq 0$ in the above LP relaxation of
the TSP. Also, $c^*$ violates the valid TSP inequality $x^*\cdot c\geq \eta$.

Suppose that the LP satisfying the Hamilton cycle property has the triangle
inequalities as constraints either explicitly or implicitly.  That is,
$$x_{i,j} \leq x_{i,k} + x_{k,j}$$ for all distinct triples $i,j,k$ of nodes (where
the order of the 2 endpoint nodes for an edge is arbitrary in the undirected
graph).
\begin{lemma}\label{splitoff}
  Suppose an LP satisfies the Hamilton cycle property and includes the
  triangle inequalities explicitly (or satisfies them at optimality).   
Then if $c'$ is certified by the LP to not be in the TSP polytope,
either the minimum cut of $c'$ is strictly less than $2$ or from $c'$ one can
construct a $c^*$ that is in the subtour relaxation of the TSP but is also
certified to not be in the TSP polytope.  
\end{lemma}
{\bf Proof: }
Suppose the minimizer for this LP \eqref{anylp} is $x^*$ and that the minimum
cut of $c'$ is at least $2$. 
Normalize $c'$ so that the minimum cut of $c'$ is exactly 2.
By assumption, $c'\cdot x^* < \eta$.  

Define
\emph{splitting off by $\eps$ at a node $k$ with nodes $i$ and $j$} to be the operation of taking nodes $i,j,k$ and
decreasing $c_{i,k}$ and $c_{j,k}$ by $\eps$ while increasing 
$c_{i,j}$ by
$\eps$ for some fixed $\eps$. In Lov\'{a}sz splitting off, this is done only
on nodes $i,j,k$ such that the minimum global cut stays the same after the
operation. 
It is known \cite{Lovaz} that if the cut $(\{k\}, V\setminus\{k\})$ is not a minimum cut,
then there exist nodes $i$ and $j$ and $\eps>0$, such that we can perform  Lov\'asz splitting
off by $\eps$ at $k$ with nodes $i$ and $j$. Thus, we iteratively perform Lov\'asz splitting
off at nodes $k$ such that $(\{k\}, V\setminus\{k\})$ is not a minimum cut, until
every node has fractional degree~2. During this process, we obtain
intermediate vectors $c''$ and eventually reach a final vector $c^*$ that has fractional degree 
$2$ at  every node.

Since Lov\'{a}sz splitting off keeps the minimum global cut the same, then $c^*$ is
a feasible subtour point.  Due to the triangle inequalities in \eqref{anylp}, 
at each step  $c^{\prime\prime}\cdot x^* < \eta$ because
$c^{\prime}\cdot x^* < \eta$ and the 
objective decreases at each step,
specifically, $c''\cdot x^* = c'\cdot x^* + \eps(x^*_{i,j} -(x^*_{i,k} + x^*_{k,j}))$. 
In the end $c^*\cdot x^* < \eta$, and
without 
loss of generality one can choose $c^*$ to be a subtour 
extreme point. \hfill \framebox \\

Now suppose the LP \eqref{anylp} is a relaxation of minimum global cut.
\begin{lemma}\label{win}
  Suppose an LP \eqref{anylp} is a relaxation of minimum global cut.  Take
  the case where it satisfies the Hamilton cycle property with $\eta =2$.    
  Then it leads to a TSP relaxation \eqref{tsprelax} that is at least as strong
  as the subtour relaxation of the TSP.
\end{lemma}
{\bf Proof: }
Suppose $c^*$ is feasible for \eqref{tsprelax} but that 
$c^*$ is not in the subtour polytope, i.e., the minimum global cut
of $c^*$ is less than $2$.  Plug $c^*$ into \eqref{anylp}.  Since LP~\eqref{anylp}
is a relaxation of minimum global cut, its minimizer $x^*$ satisfies
$c^*\cdot x^* \leq mincut(c^*) < 2$.  Thus, $y\cdot b \geq 2$ in
\eqref{tsprelax} is violated since 
$y\cdot b \leq c^*\cdot x^*$, by duality, and $c^*$ is not feasible for
\eqref{tsprelax} after all.  \hfill \framebox \\

Now suppose the LP satisfying the Hamilton cycle property is a relaxation
of minimum cut and also includes the triangle inequalities.  Then either this
LP is an exact minimum cut formulation or
it leads to a TSP relaxation that is strictly stronger than the subtour
relaxation of the TSP.
\begin{theorem}\label{winwin}
  Suppose an LP \eqref{anylp} is a relaxation of minimum global cut.
  Take the case where it
  satisfies the Hamilton cycle property
  with $\eta = 2$, and includes the triangle inequality
  constraints explicitly (or satisfies them at optimality).  
Then either this LP is an exact minimum cut formulation or
it leads to a TSP relaxation \eqref{tsprelax} that is strictly stronger than
the subtour relaxation of the TSP.
\end{theorem}
{\bf Proof: }
Lemma~\ref{win} implies that the feasible region for LP~\eqref{tsprelax} is a subset
of the subtour polytope.  In the case where \eqref{anylp} is not an exact minimum cut
formulation, we show that it is a strict subset.

Suppose \eqref{anylp} is not an exact minimum cut formulation.  Then there
exists an objective $c'$ such that $mincut(c') = 2$ but $c'\cdot x^* < 2$ for
a minimizer $x^*$ of \eqref{anylp}.  By the reasoning of
Lemma~\ref{splitoff} there exists a $c^*$ in the subtour polytope such that
$c^*\cdot x^* < 2$.  By the reasoning of Lemma~\ref{win}, $c^*$ is then not
in the relaxation \eqref{tsprelax}. \hfill\framebox \\

We provide three examples of LPs satisfying the Hamilton
cycle property here and more examples in the next section.

Consider the following linear program where $s,t\in V$:

\begin{equation}\label{hampropst}
  \begin{array}{rllll}
    \mbox{minimize }& c\cdot x \\
    \mbox{subject to } \\
    x_{i,j} &\geq& h_j - h_i &\forall ij\in E \\
    x_{i,j} &\geq& h_i - h_j &\forall ij\in E \\
    h_t - h_s &\geq& 1  \\
  \end{array}
\end{equation}
If the integrality constraint $h\in \{0,1\}$ is added to \eqref{hampropst},
the resulting integer program is an IP formulation of minimum $s,t$ cut.
Hence \eqref{hampropst} is an LP relaxation of minimum $s,t$ cut.  It is
readily shown here to satisfy the Hamilton cycle property.  
\begin{theorem}\label{hampropstthm}
  Let $c$ be the incidence vector of a Hamilton cycle.  Then the minimum
  objective value for \eqref{hampropst} is 2.
\end{theorem}

{\bf Proof: }Let $H$ be an arbitrary Hamilton cycle and $c$ be its
incidence vector.  The cycle $H$ can be decomposed into two
vertex disjoint (except at ends) $s,t$ paths $P^1$
and $P^2$.  Consider $P^1$ first.  Let the $k+1$ vertices in order in the path
$P^1$ 
starting at $s=v_0$ be $v_0,v_1,\ldots,v_k = t$.  For $i\in \{0,\ldots,k-1\}$
we have $$x_{v_iv_{i+1}} \geq h_{v_{i+1}} - h_{v_i}.$$  Summing these up over the
  edges of the path $P^1$ yields a telescoping sum 
  $$x(E(P^1)) \geq \sum_{i=0}^{k-1}(h_{v_{i+1}} - h_{v_i})= h_t-h_s\geq 1.$$
  Therefore we have $$c\cdot x = x(E(P^1)) + x(E(P^2)) \geq 2.$$
  It is easy to show the bound of $2$ is attained. \hfill \framebox \\

  If $(x^*,h^*)$ is a minimizer of \eqref{hampropst} then without loss of
  generality in what follows, we may set $$x^*_{i,j} := |h^*_j - h^*_i|.$$

  \begin{lemma}\label{triopt}
    The feasible solutions to \eqref{hampropst} (where the $x$ variables
    attain their mimimum values given the values of the $h$ variables) satisfy
    the triangle inequalities in the $x$ variables.
  \end{lemma}

  For what follows next, we introduce the concept of a {\it disjunctive
mathematical program}~\cite{balas1974disjunctive}.

  Consider the feasible regions $Q_t$ that arise from
\eqref{hampropst}
when binary integrality constraints are placed on $h$ and one sets $s:=n$ and
varies $t$ to be in the set
$\{1,\ldots,n-1\}$.  This is an IP formulation of minimum $s,t$ cut.  
Consider the disjunctive program of these $n-1$ minimum $s,t$ cut integer
programs.  The feasible region of this disjunctive program is
the convex hull of $\bigcup_t Q_t$.  
This disjunctive program can be seen to be an IP
formulation of the minimum global cut problem.  Now remove the binary
integrality constraints on the $h$ variables.  The resulting feasible regions
are $P_t$.  The feasible region $P_t$ is a relaxation of minimum $s,t$ cut for
each $t$.  
The disjunctive program of these $n-1$ LP relaxations has a feasible region
which is the convex hull of $\bigcup_t P_t$.  Hence, it is a relaxation of minimum global cut.  

\begin{theorem}
  Either $\bigcup_t P_t$ is an exact LP formulation of minimum global cut or
  it leads to a relaxation of the TSP that is strictly stronger than the
  subtour relaxation.
\end{theorem}

{\bf Proof: }The disjunctive program above is a relaxation of minimum global
cut.  Its feasible region is the convex hull of $\bigcup_t P_t$. 
 It
 satisfies the Hamilton cycle property with $\eta = 2$ by Theorem~\ref{hampropstthm}.  It also
 includes the triangle
 inequalities as constraints when at optimality by Lemma~\ref{triopt}.
 Hence the result follows from Theorem~\ref{winwin}.  \hfill \framebox \\

 This theorem intuitively says that since \eqref{hampropst}
 satisfies the Hamilton cycle property, it is likely that \eqref{hampropst}
 is an exact LP formulation of minimum $s,t$ cut.  Indeed, this turns out to
 be the case.  

Next consider the following linear program.
\begin{equation}  \label{sodasmall}
\begin{array}{lrlll}
\mbox{minimize } & c\cdot x \\
\mbox{subject to} \\
&x_{i,j} + 2w_k &\leq& x_{k,i} + x_{k,j}&\forall k<i<j \\
& x_{n-1,n} &\geq& w_{n-1} \\ 
& \sum_{k=1}^{n-1}w_k &=& 1,\\
& x,w &\geq & 0.
\end{array}
\end{equation}

\begin{theorem}
  Let $c$ be the incidence vector of a Hamilton cycle.  Then the minimum
  objective value for \eqref{sodasmall} is $2$.
\end{theorem}

{\bf Proof: }Let $c$ be the incidence vector of an arbitrary Hamilton cycle.
Let $c^0 := c$.  
For each node $k$, construct a Hamilton cycle vector $c^k$ on the set of nodes
$\{k+1,\ldots,n\}$ from a Hamilton cycle vector $c^{k-1}$ as follows.
Let $i_k$ and $j_k$ be neighbors of $k$ in the Hamilton cycle of $c^{k-1}$.
Remove these two edges incident to $k$ and replace them with the edge
$i_kj_k$.  The result of this is the Hamilton cycle vector $c^k$.  We now
derive a chain of inequalities for feasible solutions of \eqref{sodasmall}.
\begin{displaymath}
  \begin{array}{llll}
    c\cdot x &=& c^1\cdot x + x_{1i_{1}} + x_{1j_{1}} - x_{i_{1}j_{1}} \\
    &   \geq& c^1\cdot x + 2w_1 \\
    &\geq& c^2\cdot x +2w_1 +  2w_2 \\
    &\geq& \sum_k 2w_k = 2.
  \end{array}
\end{displaymath}
Thus the Hamilton cycle property is satisfied. \hfill \framebox \\

The LP \eqref{sodasmall} is a subset of the LP \eqref{soda}.  Hence
\eqref{sodasmall} is an LP relaxation of minimum cut.  However, Theorem
\ref{winwin} does not apply here unless the triangle inequalities are added to
\eqref{sodasmall}.  So whether \eqref{sodasmall} leads to a stronger
relaxation of the TSP than the subtour relaxation has not been resolved yet.

For proving the Hamilton cycle property of the $\alpha$-LP, we need the following lemma:
\begin{lemma} \label{lem:triangleInequalities}
Let $c$ be a cost vector, and let $G$ be the support graph of $c$. Consider we have the triangle constraints 
\begin{equation}
x_{i,j} \leq x_{i,k} + x_{k,j} \; \forall i,j,k \text{ s.t. } c_{k,j} > 0
\end{equation}
Let $P_{i,j}$ be a path of edges in $G$ between $i$ and $j$. Then $x(P_{i,j}) \geq x_{i,j}$ is a valid inequality.
\end{lemma}
\begin{proof}
Make the inductive assumption that the theorem holds for paths in $G$ of length $\ell$ or less. Let $P_{i,j}$ be a path in $G$ of length $\ell+1$ where $P_{i,j} =  P_{i,k} \cup \{(k,j)\}$. Then we can sum two inqualities to get the desired result:
\begin{equation}
\begin{array}{rl}
x(P_{i,k}) &\geq x_{i,k} \\
x_{i,k} + x_{k,j} &\geq x_{i,j}\\
\hline x(P_{i,j}) &\geq x_{i,j} \\
\end{array}
\end{equation}
\end{proof}

There is a symmetric optimal solution of the $\alpha$-LP when $c$ is the incidence vector of a Hamiltonian cycle. 
\begin{lemma} \label{lem:symopt}
Let $c^H$ be the incidence vector of Hamiltonian cycle $H$. There is an optimal LP solution $(x^{sym}, h^{sym})$ of the $\alpha$-LP with $x_{i,j}^{sym} = K$ for each edge $ij$ in $H$ and $h_i^{sym}= \frac{\alpha}{n}$ for each $i\in V$. 
\end{lemma}
\begin{proof}
Without loss of generality, choose the canonical Hamiltonian cycle where $c_{i,j}^H=1$ for $j=i+1$ for each $i=1,2, \dots, n-1$ and $c_{n,1}^H=1$. Let $(x^*, h^*)$ be an optimal LP solution to the $\alpha$-LP. We wish to rotate our Hamiltonian cycle by $k$ units for $k=0,1,\dots, n-1$. Define $\pi^k(i)=1+(i+k-1 \pmod n) $ for each $i\in V$. Define  $x_{i,j}^k=x^*_{\pi^k(i), \pi^k(j)}$ and $h_i^k= h^*_{\pi^k(i)}$ for each $i\neq j\in V$. By symmetry, $(x^k, h^k)$ is also an optimal LP solution. Define $x^{sym} = \frac{1}{n}\sum_{k=0}^{n-1} x^k, h^{sym} =\frac{1}{n}\sum_{k=0}^{n-1} h^k$. This is a symmetric optimal solution as desired.
\end{proof}

Now we show that the $\alpha$-LP \eqref{first-lp} satisfies the Hamilton cycle
property.
\begin{theorem} \label{thm:alphaLPHamCycle}
  Let $c$ be the incidence vector of a Hamilton cycle.  Then the minimum
  objective value for 
  the old $\alpha$-LP \eqref{first-lp} is $2$.  
\end{theorem}
    {\bf Proof sketch: }Because of the permutation symmetry of \eqref{first-lp},
    we may make a remarkable assumption.  That is, we may 
assume that the Hamilton cycle is simply going from $1$ to $n$ in order and
back to $1$ again.  But there is even another amazing consequence of this
symmetry.  We may also take a symmetric optimal solution where
$x_e = K$ for all $e$ in the canonical Hamilton cycle and $h_i= \alpha/n$ for
all nodes $i$.  
Since $x$ is an optimal solution, and 2 is an upper bound for the minimum
objective value of the old $\alpha$-LP, we know that $K \leq 2/n$.   We will show that $K$ must equal $\frac{2}{n}$,  by exploiting the fact that $x$ is a feasible solution to the old $\alpha$-LP,
thus, yielding the theorem.

By Lemma~\ref{lem:triangleInequalities} and Lemma~\ref{lem:symopt},
$x_{n,j} \leq x(P_{n,j})=jK$ where $P_{n,j}$ is the shorter path from $n$ to $j$ in the Hamiltonian cycle for each $j\in \{1,\ldots,\alpha - 1\}$.  Similarly,
$x_{n,n-j} \leq jK$ for each $j\in \{1,\ldots,\alpha - 1\}$.  Finally,
by the maxcut constraints, we have $x_{n,j} \leq h_j + h_n = \frac{2\alpha}{n}$
for each $j\in \{\alpha,\ldots,n-\alpha\}$.

Combining these, we get
\begin{equation}
  \begin{array}{lrllll}
    \sum_{j=1}^{\alpha-1}x_{n,j} &\leq& \sum_{j=1}^{\alpha-1}jK
    &=&\frac{\alpha(\alpha-1)} {2}K\\
    \sum_{j=1}^{\alpha-1}x_{n,n-j} &\leq& \sum_{j=1}^{\alpha-1}jK
    &=&\frac{\alpha(\alpha-1)}{2}K\\
    \sum_{j=\alpha}^{n-\alpha}x_{n,j} &\leq&
    \sum_{j=\alpha}^{n-\alpha}\frac{2\alpha}{n}&=&(n-2\alpha+1)\frac{2\alpha}{n}\\
    \hline \\
    x(\delta(n)) &\leq& && (n-2\alpha+1)\frac{2\alpha}{n}+
    \alpha(\alpha-1)K
  \end{array}
\end{equation}
Now,
\begin{displaymath}
  \begin{array}{lrllll}
(n-2\alpha+1)\frac{2\alpha}{ n} &=& (n-2\alpha)\frac{\alpha}{n} 
    + (n-2\alpha)\frac{\alpha}{ n} + \frac{2\alpha}{ n}\\
    &=&
    (n-2\alpha)\frac{\alpha}{n} + \alpha -\frac{2\alpha^2}{ n} +\frac{2\alpha}{ n} \\
    &=& (n-2\alpha)\frac{\alpha}{n} +\alpha -\alpha(\alpha-1)\frac{2}{n}.
  \end{array}
\end{displaymath}
Hence we obtain
\begin{align}
 x(\delta(n)) & \leq (n-2\alpha)\frac{\alpha}{n} + \alpha +\alpha(\alpha-1)\left(K-\frac{2}{n}\right) \\
    &= (n-2\alpha)h_n + \alpha+\alpha(\alpha-1)\left(K-\frac{2}{n}\right).
\end{align}
Since $K \leq 2/n$, we obtain
\begin{equation}
x(\delta(n)) \leq (n-2\alpha)h_n + \alpha.
\end{equation}
However, the degree constraint from the old $\alpha$-LP~\eqref{first-lp} says that 
$x(\delta(n)) = (n-2\alpha)h_n + \alpha,$ thus, it must be the case that $K=2/n$,
and the theorem follows.
\hfill \framebox \\

We only used the triangle inequalities $x_{i,j} \leq x_{i,k} + x_{k,j}$ when $c_{j,k} > 0$ in the proof of Theorem~\ref{thm:alphaLPHamCycle}. Thus, the old $\alpha$-LP has size $O(mn)\times O(n^2)$ as we only need a subset of the triangle inqualities.
  
\section{Bounding the objective of $\alpha$-LP on Hamiltonian Cycles}

In this section we show that when the input graph is a single cycle on all of the nodes of the graph and all edges of the graph have unit (1) capacity then the new $\alpha$-LP gives an exact solution.  That is, we show the following theorem. 

\begin{theorem}
\label{thm:ham}
  Let $c$ be the incidence vector of a Hamilton cycle. Then, the $\alpha$-LP is an exact formulation for any $\alpha\leq n/2$.
\end{theorem}

To show this for the $\alpha$-LP, we will in fact show that the Hamiltonian cycle property holds for a relaxation, called the $\beta$-LP. We relax the $\alpha$-LP by expressing $\alpha$ as the convex combination of two integers. This gives the $\beta$-LP which has the following form:

\begin{align}
\nonumber \min \quad & c \cdot x^\beta &\\
\mbox{subject to} \\
\sum_{ \{i,j\}\in E} x_{i,j}^\beta &\geq \lambda\beta_1(n-\beta_1)+(1-\lambda)\beta_2(n-\beta_2) & \label{eq:sumX}\\
\sum_{v\in V} h_v^\beta &\leq \lambda\beta_1+(1-\lambda)\beta_2& \label{eq:sumH}\\
x_{i,j}^\beta &\leq h_i^\beta+h_j^\beta &\forall \{i,j\}\in E \label{eq:maxcut}\\
x_{i,j}^\beta &\leq  x_{i,k}^\beta+x_{j,k}^\beta & \forall i,j,k\in V : c_{jk} > 0\label{eq:tri}\\
0&\leq x^\beta, h^\beta  &\label{eq:nonneg}\\
\nonumber 0&\leq \lambda \leq 1&
\end{align}

We will use  $\beta_2=\beta_1+1$ and $\alpha =\lambda \beta_1+(1-\lambda) \beta_2$ for this section. 

\subsection{Removing the $h_i$ variables}
The first thing we will consider about this linear program is how to remove the $h_i$ from the problem as they appear in very few constraints. We will project out these variables to better understand the linear program. Constraint~\eqref{eq:sumH} in the linear program guarantees that the sum of all the $h_i$ is $\alpha$. The other constraints involving $h_i$ are constraints~\eqref{eq:maxcut} of the form $x_{i,j} \leq h_i+h_j$. Let us consider that $x_{i,j}$ are fixed, then we can minimize the sum of the $h_i$ via the following linear program:

\begin{align*}
\min & \sum_{i \in V} h_i &\\
\mbox{subject to} \\
h_i+h_j &\geq x_{i,j} &\quad \forall ij\in E\\
h_i & \geq 0 & \quad \forall i \in V
\end{align*}

If we take the dual of this linear program, we get the following linear program:
\begin{align*}
\max & \sum_{ij \in E} y_{i,j}x_{i,j} &\\
\mbox{subject to} \\
\sum_{j\in V} y_{i,j} \leq 1 & \quad \forall i \in V\\
y_{i,j} & \geq 0 & \quad \forall ij \in E
\end{align*}
This is just a linear programming formulation for a 2-factor scaled down by a factor of 2. Therefore, an optimal solution to this linear program is half the cost of maximum 2-factor where $x_{i,j}$ are the edge costs. So, we could project out the $h_i$ variables, by replacing them with the constraints that every 2-factor on the $x_{i,j}$ has weight at most $2\alpha$. 

\subsection{Hamiltonian cycle property}
The Hamiltonian cycle property is a key property required of the $\beta$-LP if we are to relate it to the TSP polytope. While the $\beta$-LP is a strict relaxation of the $\alpha$-LP, it turns out the Hamiltonian cycle property still holds. Not only does the Hamiltonian cycle property hold for the $\beta$-LP, but if the $\sum_{ij} x_{i,j}$ were relaxed any further then the Hamiltonian cycle property no longer holds. 

\begin{theorem}
Consider a modified $\beta$-LP with a fixed $\lambda, \beta_1, \beta_2=\beta_1+1$ and $X$ being the right hand side of constraint~\eqref{eq:sumX} and everything else the same. Call this modified LP the $(\beta, X)$-LP. The Hamiltonian cycle property holds for this LP if and only if 
\[X\geq \lambda \beta_1(n-\beta_1) +(1-\lambda)\beta_2(n-\beta_2).\]
\end{theorem}

\begin{proof}
Without loss of generality, consider the Hamiltonian cycle is $1,2,\dots n$. Let $(x,h)$ be a feasible solution to the $(\beta, X)$-LP. By symmetry, we can consider all the cyclic rotations of $1,2,\dots n$ on $(x,h)$ and these are all feasible. Similarly, by taking the convex combination of these rotations where each is weighted by $\frac{1}{n}$ then we also arrive at a feasible solution. Call this modified solution $(x', h')$. 

In the modified solution, $h'_i = \frac{\alpha}{n}$ for all $i\in V$ by symmetry. Similarly, $x'_{i,j}=x'_{i+1, j+1}$ for all $i,j\in V$ by symmetry. Now let $K=x'_{1,2}$.  We will compute the maximum possible sum of the $x'_{i,j}$. Let $d_n(i,j)$ denote the distance between nodes $i$ and $j$ in the Hamiltonian cycle $1,2,\dots n$. By the triangle inequality, given by constraint~\eqref{eq:tri}, we have the following bound:
\[ x'_{i,j} \leq K d_n(i,j)\]
By the max cut inequality, given by constraint~\eqref{eq:maxcut}, we have the following upper bound:
\[ x'_{i,j} \leq \frac{2\alpha}{n}\]
When $d_n(i,j) \leq \beta_1$ we will use the first upper bound and otherwise we will use the second upper bound. Combining these upper bounds we get:
\begin{align*}
\sum_{ij\in E} x'_{i,j} &= \sum_{ij \in E, d_n(i,j) \leq \beta_1} x'_{i,j} +\sum_{ij \in E, d_n(i,j) > \beta_1}x'_{i,j}\\
&\leq \sum_{ij \in E, d_n(i,j) \leq \beta_1} Kd_n(i,j) +\sum_{ij \in E, d_n(i,j) > \beta_1}\frac{2\alpha}{n}\\
&\leq \sum_{i=1}^{\beta_1} nKi + \left(\frac{n(n-1)}{2}-\beta_1 n \right) \frac{2\alpha}{n}\\
&\leq nK\frac{\beta_1(\beta_1+1)}{2} + 2\alpha \left(\frac{(n-1)}{2}-\beta_1  \right)\\
\end{align*}
We know that the sum of the $x'_{i,j}$ is $X$. Let $X=\lambda \beta_1(n-\beta_1) + (1-\lambda)\beta_2(n-\beta_2)$. We can first simplify this expression to:
\begin{align*}
\sum_{ij} x'_{i,j} &= \lambda \beta_1(n-\beta_1) + (1-\lambda)\beta_2(n-\beta_2)\\
&= \lambda \beta_1(n-\beta_1) + (1-\lambda)\left(\beta_1(n-\beta_1)+(n-2\beta_1-1)\right)\\
&= \beta_1(n-\beta_1)+(1-\lambda)(n-2\beta_1-1)\\
\end{align*}
Now we will use this to find a bound on $K$:
\begin{align*}
\beta_1(n-\beta_1)+(1-\lambda)(n-2\beta_1-1) &\leq nK\frac{\beta_1(\beta_1+1)}{2} + 2\alpha \left(\frac{(n-1)}{2}-\beta_1  \right)\\
\beta_1(n-\beta_1)+(1-\lambda)(n-2\beta_1-1) &\leq \frac{n}{2}K \beta_1(\beta_1+1)+ (\beta_1+1-\lambda)(n-2\beta_1-1)\\
\beta_1(n-\beta_1) &\leq \frac{n}{2} K\beta_1(\beta_1+1)+ \beta_1(n-2\beta_1-1)\\
\beta_1(\beta_1+1) &\leq \frac{n}{2} K\beta_1(\beta_1+1)\\
\frac{2}{n} &\leq K
\end{align*}
Therefore, the Hamiltonian cycle property holds when $X$ is at least $\lambda\beta_1(n-\beta_1)+(1-\lambda)\beta_2(n-\beta_2)$. 

Now let $X=\lambda \beta_1(n-\beta_1) + (1-\lambda)\beta_2(n-\beta_2)-\epsilon$, and consider  the following assignment of $h_i, x_{i,j}$. 
\begin{align*}
h_i &= \frac{\alpha}{n} \\
x_{i,j} &= \begin{cases}
(\frac{2}{n}-\frac{2\epsilon}{n\beta_1(\beta_1+1)})d_n(i,j) & \text{ if } d_n(i,j) \leq \beta_1\\
\frac{2\alpha}{n} & \text{ otherwise}
\end{cases}
\end{align*}
This is a valid solution for the $(\beta, X)$-LP with the given $X$. This solution gives value less than 2 for $c\cdot x$ where $c$ is the Hamiltonian cycle $1,2,\dots, n$. Therefore, we have proven that $X=\lambda \beta_1(n-\beta_1) + (1-\lambda)\beta_2(n-\beta_2)$ is the smallest value of $X$ such that $(\beta, X)$ satisfies the Hamiltonian cycle property. 
\end{proof}
This tells us that the $\beta$-LP does have the Hamiltonian cycle property, and is the most relaxed such LP given constraints~\eqref{eq:maxcut}, \eqref{eq:tri}, \eqref{eq:nonneg} still hold. 
\begin{corollary}
The Hamiltonian cycle property holds for the $\beta$-LP when $\beta_2=\beta_1+1$. 
\end{corollary}
\begin{proof}
The sum of the $x_{i,j}$ is always $\lambda \beta_1(n-\beta_1)+(1-\lambda)\beta_2(n-\beta_2)$ in the $\beta$-LP. Therefore, the previous theorem holds. 
\end{proof}

\begin{proof}[Proof of ~\ref{thm:ham}]
By making $\lambda$ a constant equal to $1$ in the $\beta$-LP, then we get back the $\alpha$-LP exactly. So, the Hamiltonian cycle property holds for the $\alpha$-LP as well with value $2$. 

The solution $h^\alpha_i= \alpha/n$ and $x^\alpha_{i,j} = \min(\frac{2|j-i|\mod n}{n}, \frac{2\alpha}{n})$ provides a solution with cost exactly $2$. 
\end{proof}

Interestingly, this last property of the $x^{\alpha}_{i,j}$ motivates the use of the single minimum cut constraint in the $\alpha$-LP. In particular, 
\[ \alpha(n-\alpha) = \sum_{ij\in E} \min(\frac{2|j-i|\mod n}{n}, \frac{2\alpha}{n})\].


 \section{Using the $\alpha$-LP to strengthen a  TSP relaxation}
 
For any $\alpha $ the $\alpha$-LP achieves the optimal solution on the Hamiltonian cycle when $\lambda = 0$ as stated in Theorem \ref{thm:ham}. This leads to an interesting relaxation of the TSP using the technique relating 
compact separation to compact optimization \cite{Martin91,CarrL04}.  Let 
$c^{obj}$ be the TSP objective function and $t$ be the TSP edge variables.  Recall 
the $\alpha$-LP's objective is $c\cdot x$.    
Consider the dual of the $\alpha$-LP with $t := c$ for the columns in $A$  
corresponding to the $x$ variables and $t := 0$ otherwise, which is 
\begin{displaymath}
\begin{array}{lll}
\mbox{maximize }& y^{\alpha}\cdot b^{\alpha}\\
\mbox{subject to} \\
&y^{\alpha}\cdot A^{\alpha} \leq t \\
&y^{\alpha} \geq 0.
\end{array}
\end{displaymath}
Note that $c^{obj}$ and $t$ are set to $0$ on the columns corresponding to the 
$h$ variables.  Then the $\alpha$-TSP relaxation is     
\begin{displaymath}
\begin{array}{lll}
\mbox{minimize }&c^{obj}\cdot t \\
\mbox{subject to} \\
&y^{\alpha}\cdot A^{\alpha} \leq t \\
&y^{\alpha}\cdot b^{\alpha} \geq 2 \\
&y^{\alpha} \geq 0.
\end{array}
\end{displaymath}
How this compares to the subtour relaxation of the TSP is of interest.  
Compact subtour relaxations can be found in \cite{Arthanari,Carr96,Martin91}. 
We have the following surprising result:
\begin{lemma}
If the $\alpha$-LP gives an answer of strictly less than $2$ for a 
subtour extreme point $c^*$, then $\alpha$-TSP can be made stronger than the 
subtour relaxation (after adding a compact subtour relaxation to it).
\end{lemma}
{\bf Proof: }One can see that $t := c^*$ and $t\cdot x^* < 2$ for a feasible 
$\alpha$-LP solution $x^*$ is contradicted by weak duality and 
$y^{\alpha}\cdot b^{\alpha} \geq 2$.  Therefore, $c^*$ is an infeasible value 
for $t$.  \hfill \framebox \\

The companion theorem that uses Lov\'{a}sz splitting \cite{Lovaz} is: 
\begin{lemma}
If the $\alpha$-LP gives an answer of strictly less than the global minimum cut for some 
cost function $c^{\prime}$, then the $\alpha$-LP gives an answer of strictly less 
than $2$ for some subtour extreme point $c^*$.
\end{lemma}

Combining these two lemmas gives us the following theorem: 
\begin{theorem}
Either $\alpha$-LP always gives an answer of at least the minimum global cut 
or we have produced, using $\alpha$-TSP, a compact relaxation stronger than 
the subtour relaxation.
\end{theorem}

\section{Minimum Spanning Trees}
In this subsection, we will examine the cost of the minimum spanning tree (MST) of a point $x$ which has the Hamiltonian cycle property to help determine if $x$ dominates a convex combination of cuts. We say a point satisfies the Hamiltonian cycle property if $c\cdot x\geq 2$ for any cost function $c$ which is $1$ on the edges of Hamiltonian cycle and 0 elsewhere and $x$ is a metric. The cost of the minimum Hamiltonian cycle is bounded above by twice the MST, so the value of the MST on $x$ is at least $1$. For certain values of the MST we can determine that $x$ dominates a convex combination of cuts. In particular, when the value of the MST is exactly 1, or at least 2. 

\begin{lemma}
Let the Hamiltonian cycle property hold for $x$ and let the cost of the minimum spanning tree be exactly 1, then $x$ is exactly a convex combination of cuts. 
\end{lemma}
\begin{proof}
Let $x$ have a minimum spanning tree of cost exactly $1$, and call this MST $T$. Let $P_T(i,j)$ be the path in $T$ from $i$ to $j$.  Now consider there is any edge $ij$. Originally, we can use $2T$ to make a Hamiltonian cycle of value $2$ (via shortcutting). Now $T'=2T+ij-P_T(i,j)$ is Eulerian and connected, therefore we can use this as a Hamiltonian cycle. We know the Hamiltonian cycle property holds for $x$, it must be the case $x_{i,j} \geq x(P_T(i,j))$. Otherwise the Hamiltonian cycle $T'$ gives a cost function $c_{T'}$ where $c_{T'}x <2$. We also, know that $x$ is a metric though which gives $x_{i,j} \leq x(P_T(i,j))$. Therefore, we know that $x_{i,j} = x(P_T(i,j))$ for all $i,j$ when the MST is exactly $1$. Let $S_{i,j}$ be the cut induced by the components of $T-ij$. We can write $x$ as a convex combination of cuts by writing it as $\sum_{ij\in T} x_{i,j} \delta(S_{i,j})$. 
\end{proof}

\begin{lemma}
Let the Hamiltonian cycle property hold for $x$ and let the cost of the minimum spanning tree be at least 2, then $x$ dominates a convex combination of cuts. 
\end{lemma}

\begin{proof}
We know that the natural Steiner tree LP has a integrality gap of at most 2. Consider the special case when all nodes are terminals, and we get the following LP and its dual
\begin{displaymath}
\begin{array}{llll}
\mbox{minimize}&x\cdot y \\
\mbox{subject to} 
& y(\delta(S)) \geq 1 & \forall S \subseteq V, S \neq \emptyset, V\\
& y_{i,j} \geq 0 & \forall ij \in E(V)
\end{array}
\end{displaymath}
\begin{displaymath}
\begin{array}{llll}
\mbox{maximize}&1\cdot z \\
\mbox{subject to} 
& \sum_{S:ij \in \delta(S)} z_S \leq x_{i,j} & \forall ij \in E(V) \\
& y_{S} \geq 0 & \forall S \subseteq V, S \neq \emptyset, V\\
\end{array}
\end{displaymath}
Given the minimum spanning tree has cost 2, then there is a solution to the dual linear program of cost at least 1. The $y_S$ variables give a decomposition of cuts that $x$ dominates. Therefore, $x$ dominates a convex combination of cuts as desired. 

\end{proof}

\section{Removing Knowledge of $\alpha$}

Recall that $z^{\alpha} = [x^{\alpha},h^{\alpha}]$ in $\alpha$-LP.  
We can now use the ideas of lift-and-project \cite{balas1994mixed} to obtain our first 
minimum global cut LP formulation that has symmetricity: 
\[
\begin{array}{lllll}
\mbox{minimize} & c\cdot x \\
\mbox{subject to} \\
&A^{\alpha}z^{\alpha} \geq b^{\alpha}\lambda_{\alpha} \\
&z^{\alpha} \geq 0 \\
&z = \sum_{\alpha}z^{\alpha}& \\
&\sum_{\alpha}\lambda_{\alpha} = 1, \, \lambda \geq 0.
\end{array}
\]
This $O(mn^2) \times O(n^3)$ LP is at least one order of magnitude bigger than the 
minimum cut LP ($O(n^3) \times O(n^2)$) in  SODA~\cite{CarrKLNP07, carr2009compacting}, which did not
have symmetricity.  We can however shed almost all of this 
extra size if we sacrifice exactness by a small amount.  

\subsection{The $\gamma$-LP}

Previously, we chose $\beta_2= \beta_1+1$. This worked well, but we can relax this further. Let us consider the $\beta$-LP with $\beta_2=2\beta_1$ and call this the $\gamma$-LP. We will show that any solution to the $\gamma$-LP is not too far from a solution to the $\alpha$-LP. With the following theorem, we can solve the $\gamma$-LP with all the power of $2$ values for $\beta_1$, and get an approximate solution to the $\alpha$-LP.

\begin{theorem}
\label{thm:alpobj}
Consider any solution $(x, h, \lambda)$ to the $\gamma$-LP with $\beta_2=2\beta_1$. Then $\frac{9}{8}(x, h)$ is a solution to the $\alpha$-LP for some value of $\alpha$. 
\end{theorem}

\begin{proof}
We will be considering a relaxation of the $\alpha$-LP where we will drop the constraint that variables be bounded above by $1$.

First, it is clear that Constraints~\eqref{constalp:three},~\eqref{constalp:four}, and the non-negativity constraints are satisfied by $\frac{9}{8}(x, h)$ as the constraints are homogeneous and were satisfied by $(x,h, \lambda)$ in the $\gamma$-LP.  So, we only need to show that the remaining constraints hold. 

Let $X=\sum_{ij\in E}x_{i,j}$ and $H=\sum_{i=1}^n h_i$. Based on constraints~\eqref{eq:sumX}, and~\eqref{eq:sumH}, then we get that:
\begin{align*}
H&=\beta_1\lambda+(1-\lambda)\beta_2 \\
 &= 2\beta_1-\lambda \beta_1 \\
X&= \lambda \beta_1(n-\beta_1) + (1-\lambda) \beta_2(n-\beta_2)\\
 &= \beta_1(n-\beta_1)+(1-\lambda)\beta_1 (n-3\beta_1)\\
X&= H(n-3\beta_1)+2\beta_1^2  
\end{align*}
Now to satisfy the $\alpha$-LP the relationship between $\frac{9}{8}X=X'$ and $\frac{9}{8}H=H'$ must be:
\[
X' \geq H' (n-H')
\]
So, we want to show that:
\[
\frac{9}{8} (H(n-3\beta_1)+2\beta_1^2) \geq \frac{9}{8} H (n-\frac{9}{8} H).
\]
Rewriting the above expression we get:
\[
\frac{9}{8} H^2 - 3\beta_1 H +2\beta_1^2 \geq 0.
\] 
The left hand side is a quadratic which has it's minimum at $H=\frac{4\beta_1}{3}$. Plugging in this value of $H$, we see the minimum value of the quadratic is $0$ as desired. By letting, $\alpha = H'=\frac{9}{8}H$ then we have also shown that constraints~\eqref{constalp:two} and~\eqref{constalp:one} are satisfied by $\frac{9}{8}(x,h)$. 

The above works as long as $H \leq \frac{4}{9}n$. For the case where $H>\frac{4}{9}n$, then we consider $\beta_1=n/4$. Plugging this into the equation relating $X$ and $H$ we get:
\[
X = H(n-3\beta_1)+2\beta_1^2  =H\frac{n}{4}+\frac{n^2}{8} 
\]
Now, if we scale by $c=\frac{n}{2H}$ we get $H'=cH=\frac{n}{2}$, and we now want 
\[
X' \geq H'(n-H')
\]
Plugging in the above variables we get:
\begin{align*}
X' &= cH\frac{n}{4}+\frac{n^2}{8} \\
&= \frac{n}{2}\cdot\frac{n}{4}+\frac{n^2}{8} \\
&=\frac{n^2}{4} \\
&= \frac{n}{2} \cdot \frac{n}{2} \\
&= H' (n-H')
\end{align*}
So, we have the desired result in this case. 

When $H$ is less than $\frac{4}{9}n$, then $\frac{9}{8}(x,h)$ is a solution the $\alpha$-LP for some $\alpha\leq \frac{n}{2}$. When $H \geq \frac{4}{9}n$, then $\frac{n}{2H}(x,h)$ is a solution to the $\alpha$-LP for $\alpha= \frac{n}{2}$.
\end{proof}

\begin{corollary}
\label{coro}
Either $\gamma$-LP always gives an answer of at least 
$\frac{8}{9}$ times the minimum global cut 
or we have produced, using $\gamma$-TSP, a compact relaxation stronger than 
the subtour relaxation.
\end{corollary}

Let $0<\eps\leq1$.  Partition the interval from 1 to $\frac{n}{2}$ into $q = \lceil{\log_{1+\eps}\frac{n}{2}}\rceil$ sets
with boundary points at $\{1,p_1,p_2,...,p_{q-1},\frac{n}{2}\}$ where $p_\gamma = (1+\eps)^{\gamma}$.
For the $\gamma^{th}$ partition, $\beta_1 = (1+\eps)^{\gamma-1}$ and $\beta_2 = (1+\eps)^{\gamma} = (1+\eps)\beta_1$.
Using disjunctive programming on the $q$ linear programs, we obtain an approximate minimum cut LP with symmetry. 
\begin{displaymath}
\begin{array}{lllll}
\mbox{minimize} & c\cdot x \\
\mbox{subject to} \\
&A^{\gamma}z^{\gamma} \geq b^{\gamma}\lambda_{\gamma} \\
&z^{\gamma} \geq 0 \\
&z = \sum_{\gamma}z^{\gamma} \\
&\sum_{\gamma}\lambda_{\gamma} = 1, \, \lambda \geq 0.
\end{array}
\end{displaymath}

Note that when $\eps=1$, the disjunction of $\gamma$-LPs has size $O(mn\log{n})\times O(n^2 \log n)$ and integrality gap of $\frac{9}{8}$. If the $\gamma$-LP cannot be used to strengthen the subtour elimination LP for the TSP (by Corollary \ref{coro}), then
this outperforms all other known minimum global cut LPs when $m = O(n^k)$, $1<k<2$.
\section{Conclusion}

In this paper, we proved that a linear program relaxation of the minimum cut problem with the Hamiltonian cycle property 
has the minimum cut as its lower bound,
or its dual can be used to strengthen the subtour elimination linear program for the TSP. In addition, we've also shown that the $\alpha$-LP and $\beta$-LP both satisfy the Hamiltonian cycle property. So, either these linear programs
are smaller than the standard programs
used to compute minimum cuts, or they can be used to get a strengthening of the subtour elimination linear program. In addition, we provided evidence that in some special cases, based on the structure of the solution, the
$\alpha$-LP 
has the minimum cut as a lower bound.
In addition to these linear programs, we can extend the $\beta$-LP to the $\gamma$-LP. Even though the $\gamma$-LP, with $\eps=1$, is only an $\frac{8}{9}$-approximation to the $\beta$-LP, we only need to solve $O(\log n)$ instead of $O(n)$ linear programs in the disjunction of the $\gamma$-LPs for this approximation. This makes this disjunctive program the smallest linear program for the minimum global cut with a gap of $\frac{9}{8}$ or better, given that the program does not strengthen the subtour elimination LP for the TSP.

The big remaining question is whether or not the $\alpha$-LP and $\beta$-LP both have the minimum cut as a lower bound,
or whether they can be used to strengthen the subtour elimination linear program. 
\bibliographystyle{elsarticle-harv}
\bibliography{mincut}

\end{document}